\documentclass[12pt]{article}

\usepackage{amsmath,amssymb,amstext,amsthm}

\theoremstyle{plain}
\newtheorem{theorem}{Theorem}
\newtheorem{lemma}{Lemma}
\newtheorem{proposition}[lemma]{Proposition}

\theoremstyle{definition}
\newtheorem*{definition}{Definition}

\renewcommand{\le}{\leqslant}
\renewcommand{\ge}{\geqslant}
\newcommand{\Oh}{\mathcal{O}}
\renewcommand{\Pr}[1]{\ensuremath{\operatorname{\mathbf{Pr}}\left[#1\right]}}
\newcommand{\Ex}[1]{\ensuremath{\operatorname{\mathbf{E}}\left[#1\right]}}
\newcommand{\Var}[1]{\ensuremath{\operatorname{\mathbf{Var}}\left[#1\right]}}
\newcommand{\diam}{\operatorname{diam}}
\newcommand{\polya}{\operatorname{Urn}}
\newcommand{\na}{N^{\ast}}

\title{Randomized Rumor Spreading in \\ Poorly Connected Small-World Networks\thanks{A preliminary version of this paper has appeared in proceedings of the 28th International Symposium on Distributed Computing (DISC 2014), pages 346--360.}}

\author{Abbas Mehrabian\thanks{Supported by the Vanier Canada Graduate Scholarships program. Part of this work was done while the first author was visiting Monash University, Australia.} \\
{\small Department of Combinatorics and Optimization,}\\ {\small University of Waterloo, Waterloo, Ontario, Canada} \\
{\small \texttt{amehrabi@uwaterloo.ca}}
\and  Ali Pourmiri \\
{\small Max Planck Institute for Informatics, Saarbr\"ucken, Germany} \\
{\small \texttt{pourmiri@mpi-inf.mpg.de}}
}

\begin{document}
\maketitle

\begin{abstract}
The {\sf Push-Pull} protocol is a well-studied round-robin rumor spreading protocol defined as follows: initially a node knows a rumor and wants to spread it to all nodes in a network quickly.
In each round, every informed node sends the rumor to a random neighbor, and every uninformed node contacts 
a random neighbor and gets the rumor from her if she knows it.
We analyze the behavior of this protocol on random $k$-trees, a class of power law graphs, which are small-world and have large clustering coefficients, built as follows:
initially we have a $k$-clique.
In every step a new node is born, a random $k$-clique of the current graph is chosen,
and the new node is joined to all nodes of the $k$-clique.
When $k\ge2$ is fixed, we show that if initially a random  node is aware of the rumor,
then with probability $1-o(1)$ after 
$\Oh\left( (\log n)^{1+{2}/{k}} \cdot \log \log n\cdot f(n) \right)$rounds the rumor propagates to $n-o(n)$ nodes,
where $n$ is the number of nodes and $f(n)$ is any slowly growing function.
Since these graphs have polynomially small conductance, vertex expansion $\Oh(1/n)$ and constant treewidth,
these results demonstrate that {\sf Push-Pull} can be efficient even on poorly connected  networks.

On the negative side, we prove that  with probability $1-o(1)$ the protocol needs at least $\Omega\left(n^{({k-1})/({k^2+k-1})}/f^2(n)\right)$ rounds to inform all nodes. 
This exponential dichotomy between time required for informing \emph{almost all} and \emph{all} nodes is striking. 
Our main contribution is to present, for the first time, a natural class of random graphs in which such a phenomenon can be observed.
Our technique for proving the upper bound successfully carries over to a closely related class of graphs, the random $k$-Apollonian networks, for which we prove an upper bound of 
$\Oh\left( (\log n) ^{c_k} \cdot \log \log n \cdot f(n) \right)$ rounds for informing $n-o(n)$ nodes with probability $1-o(1)$ when $k\ge3$ is fixed.
{Here, $c_k = (k^2-3)/(k-1)^2 < 1 + 2/k$.}

\textbf{Keywords:} randomized rumor spreading,
push-pull protocol,
random $k$-trees,
random $k$-Apollonian networks,
{urn models}.
\end{abstract}

\section{Introduction}

Randomized rumor spreading is an important primitive for information dissemination in networks and has
numerous applications in network science, ranging from spreading information in the WWW and Twitter to spreading viruses and diffusion of ideas in human communities (see~\cite{CLP09,DFF11,DFF12,DFF12b,FPS12}).
A well studied rumor spreading protocol is the \emph{\textsf{Push-Pull} protocol},
introduced by Demers, Greene, Hauser, Irish, Larson, Shenker, Sturgis, Swinehart, and Terry~\cite{DGH+87}.
Suppose that one node in a  network is aware of a piece of information, the `rumor.'
The protocol proceeds in rounds.
In each round, every \emph{informed} node contacts  a random neighbor and sends the rumor to it (`pushes' the rumor),
and every \emph{uninformed} nodes contacts a random neighbor and gets the rumor if the neighbor knows it (`pulls' the rumor).
Note that this is a synchronous protocol, e.g.\ a node that receives a rumor in a certain round cannot send it on in the same round.

A point to point communication network can be modelled as an undirected graph: the nodes represent the processors and the links represent communication channels between them. 
Studying rumor spreading has several applications to distributed computing in such networks, of which we mention just two. 
The first is in broadcasting algorithms: a single processor wants to broadcast a piece of information to all other processors in the network (see~\cite{broadcast_survey} for a survey). 
There are at least three advantages to 
the {\textsf{Push-Pull} protocol}:
it is simple (each node makes a simple local decision in each round; no knowledge of the global topology is needed; no state is maintained), scalable (the protocol is independent of the size of network: it does not grow more complex as the network grows) and robust (the protocol tolerates random node/link failures without the use of error recovery mechanisms, see~\cite{FPRU90}).
A second application comes from the maintenance of
databases replicated at many sites, e.g., yellow pages, name servers, or server directories. 
There are updates injected at various nodes, and these updates must propagate to all nodes in the network. 
In each round, a processor communicates with a random neighbor and they share any new information, so that eventually all copies of the database converge to the same contents. See~\cite{DGH+87} for details.
Other than the aforementioned applications, rumor spreading protocols have successfully been applied in various contexts such as resource discovery~\cite{Harchol-Balter1999}, distributed averaging~\cite{Boyd2006}, data aggregation~\cite{KDG03},
and the spread of computer viruses~\cite{BBCS05}.

We only consider simple, undirected and connected graphs.
For a graph $G$, let $\Delta(G)$ and $\diam(G)$ denote the maximum degree and the diameter of $G$, respectively, and let $\deg(v)$ denote the degree of a vertex $v$.
Most studies in randomized rumor spreading focus on the \emph{runtime} of this protocol,
defined as the number of rounds taken until a rumor initiated by one vertex reaches all other vertices.
It is clear that $\diam(G)/2$ is a lower bound for the runtime of 
the {\textsf{Push-Pull} protocol}.
We say an event happens \emph{with high probability (whp)} if its probability approaches $1$ as $n$ goes to infinity.
Feige, Peleg, Raghavan and Upfal~\cite{FPRU90} showed that for an $n$-vertex $G$, whp the rumor reaches all vertices in
$\Oh(\Delta(G)\cdot (\diam(G)+\log n))$ rounds.
This protocol has been studied on many graph classes such as
complete graphs~\cite{KSSV00},
Erd\H{o}s-R\'{e}yni random graphs~\cite{comcom,FPRU90,FHP10,finest}, random regular graphs~\cite{BEF08,FP10}, and hypercube graphs~\cite{FPRU90}.
For most of these classes it turns out that whp the runtime is $\Oh(\diam(G)+\log n)$, which does not depend on the maximum degree.

Randomized rumor spreading has recently been studied on real-world networks models.
Doerr, Fouz, and Friedrich~\cite{DFF11} 
proved an upper bound of $\Oh(\log n)$ for the runtime on preferential attachment graphs, and
Fountoulakis, Panagiotou, and Sauerwald~\cite{FPS12} proved the same upper bound (up to constant factors) for the runtime on {the giant component of} random graphs with given expected degrees (also known as the Chung-Lu model) with power law degree distribution.

The runtime  is closely related to the \emph{expansion profile} of the graph.
Let $\Phi(G)$ and $\alpha(G)$ denote the conductance
and the vertex expansion
of a graph $G$, respectively.
After a series of results by various scholars, Giakkoupis~\cite{G13,Gia11} showed that for any \mbox{$n$-vertex} graph $G$, the runtime of the {\sf Push-Pull} protocol is
${\Oh}\left(\min\{\Phi(G)^{-1}\cdot{\log n}, \alpha(G)^{-1}\cdot \log \Delta(G) \cdot \log n \}\right)$.
It is known  that  whp preferential attachment graphs and random graphs with {power law} expected degrees have
conductance $\Omega(1)$ (see~\cite{CLV03,MPS03}).
So it is not surprising that rumors spread fast on these graphs.
Censor-Hillel, Haeupler, Kelner, and Maymounkov~\cite{CHKM12} presented a different rumor spreading protocol that whp distributes the rumor in $\Oh(\diam(G)+{\mathrm{polylog}}(n))$ rounds on any connected $n$-vertex graph, which seems particularly suitable for poorly connected graphs.

\subsection{Our contribution}

We study the {\sf Push-Pull}  protocol on random $k$-trees,
a class of random graphs defined as follows.
\begin{definition}[Random $k$-tree process~\cite{Ga09}]\label{def_random_k_tree}
Let $k$ be a positive integer.
Build a sequence $G(0)$, $G(1),$ $\dots$ of random graphs as follows.
The graph $G(0)$ is just a clique on $k$ vertices.
For each $1\leq t\leq n$, $G(t)$ is obtained from $G(t-1)$ as follows:
a $k$-clique of $G(t-1)$ is chosen uniformly at random,
a new vertex is born and is joined to all vertices of the chosen $k$-clique.
The graph $G(n)$ is called a random $k$-tree on $n+k$ vertices.
\end{definition}
We remark that this process is different from the random $k$-tree process defined by Cooper and Uehara~\cite{CU10} which was further studied in~\cite{CF13}.

Sometimes it is convenient to view this as a `random graph evolving in time.'
In this interpretation, in every round $1,2,\dots,$ a new vertex is born and is added to the evolving graph,
and $G(t)$ denotes the graph at the end of round $t$.
Observe that $G(t)$ has $k+t$ many vertices and $kt+1$ many $k$-cliques.

As in the preferential attachment scheme,
the random $k$-tree process enjoys a `the rich get richer' effect.
Think of the number of $k$-cliques containing any vertex $v$ as the `wealth' of $v$ (note that this quantity is linearly related to $\deg(v)$). 
Then, the probability that the new vertex attaches to $v$ is proportional to the wealth of $v$,
and if this happens, the wealth of $v$ increases by $k-1$.
On the other hand, random $k$-trees have much larger clustering coefficients than preferential attachment graphs, as all neighbors of each new vertex are joined to each other.
It is well-known that real-world networks tend to have large clustering coefficients (see, e.g.,~\cite[Table~1]{WS98}).
  
Gao~\cite{Ga09} showed that whp the degree sequence of $G(n)$ asymptotically follows a power law distribution with exponent $2+\frac{1}{k-1}$.
In Section~\ref{expansion} we show that whp the diameter of $G(n)$ is $\Oh(\log n)$, and its clustering coefficient is at least 1/2, as opposed to preferential attachment graphs and random graphs with {power law} expected degrees, whose clustering coefficients are $o(1)$ whp.
As per these properties, random $k$-trees may serve as more realistic models for real-world networks.

On the other hand, in Section~\ref{expansion} we prove that with high probability a random $k$-tree on $n+k$ vertices
has conductance $\Oh\left(\log n \cdot n^{-1/k}\right)$ and vertex expansion $\Oh(k/n)$. Therefore we cannot resort to existing results linking the runtime to expansion properties
to show rumors spread fast in these graphs.
Another interesting structural property of a random $k$-tree is its {\it treewidth}
(see~\cite{Kl94} for a comprehensive survey).
Gao~\cite{Ga12} proved that many random graph models, including Erd\H{o}s-R\'{e}yni random graphs with expected degree $\omega(\log n)$ and preferential attachment graphs with out-degree greater than $11$, have treewidth $\Theta(n)$,
whereas all random $k$-trees have treewidth $k$ by construction. (According to~\cite{Ga12}, not much is known about the treewidth of a preferential attachment graph with out-degree between 3 and 11.)

In conclusion, distinguishing features of random $k$-trees, such as high clustering coefficient, {bad expansion (polynomially small conductance) and tree-like structure (small treewidth)},
inspired us to study randomized rumor spreading on this unexplored random environment.
Our first main contribution is the following theorem.

\begin{theorem}
\label{thm:mainrandomktree}
Let $k \ge2$ be constant and let $f(n)=o(\log \log n)$ be an arbitrary function going to infinity with $n$.
If initially a random vertex of an $(n+k)$-vertex random $k$-tree knows a rumor, then
with high probability after $\Oh\left( (\log n)^{1+\frac{2}{k}} \cdot \log \log n\cdot f(n)^{\frac{3}{k}} \right)$ rounds of the {\sf Push-Pull} protocol, $n-o(n)$ vertices will know the rumor.
\end{theorem}

We give a high-level sketch of the proof of Theorem~\ref{thm:mainrandomktree}.
Let $m=o(n)$ be a suitably chosen parameter,
and note that $G(m)$ is a subgraph of $G=G(n)$.
Consider the connected components of $G - G(m)$.
Most vertices born later than round $m$ have relatively small degree,
so most these components have  a small maximum degree (and logarithmic diameter)
thus the rumor spreads quickly inside each of them.
A vertex $v \in V(G(m))$ typically has a large degree, but this means there is a high chance
that $v$ has a neighbor $x$ with small degree, which quickly receives the rumor from $v$ and spreads it (or vice versa).
We  build an almost-spanning tree $T$ of $G(m)$ with logarithmic height, such that
for every edge $uv$ of $T$, one of $u$ and $v$ have a small degree,
or $u$ and $v$ have a common neighbor with a small degree.
Either of these situations mean the rumor is exchanged quickly between $u$ and $v$.
This tree $T$ then works as a `highway system' to spread the rumor within vertices of $G(m)$
and from them to the components of $G-G(m)$.

The main novelty in this proof is how the almost-spanning tree is built and used (using small degree vertices for fast rumor transmission between high degree vertices has also been used in previous papers, e.g.\ \cite{DFF11,FPS12}).
Our second main contribution is the following theorem, which gives a polynomial lower bound for the runtime.

\begin{theorem}
\label{thm:lowerbound}
Let $f(n)=o(\log \log n)$ be an arbitrary function going to infinity with $n$.
Suppose that initially one vertex in the random $k$-tree, $G(n)$, knows the rumor. 
Then, with high probability the {\sf Push-Pull} protocol needs at least 
$n^{(k-1)/(k^2+k-1)}f(n)^{-2}$ rounds to inform all vertices of $G(n)$.
\end{theorem}

We give a high-level sketch of the proof of Theorem~\ref{thm:lowerbound}. 
A \emph{barrier} in a graph is a subset $D$ of edges of size $\Oh(1)$, whose deletion disconnects the graph.
If both endpoints of every edge of a barrier $D$ have  large degrees, then the protocol needs a  large time to pass the rumor through $D$.
For proving Theorem~\ref{thm:lowerbound}, we prove a random $k$-tree has a barrier whp.
The main novelty in this proof is introducing and using the notion of a barrier.

It is instructive to contrast Theorems~\ref{thm:mainrandomktree} and~\ref{thm:lowerbound}.
The former implies that if you want to inform almost all the vertices, then you just need to wait for a polylogarithmic number of rounds.
The latter implies that, however, if you want to inform each and every vertex, then you have to wait for polynomially many rounds. 
This is a striking phenomenon and the main message of this paper is to present, for the first time, a natural class of random graphs in which this phenomenon can be observed.
In fact, in applications such as viral marketing and voting, it is more appealing to inform 99 percent of the vertices very quickly instead of waiting a long time until everyone gets informed.
For such applications, Theorem~\ref{thm:mainrandomktree} implies that the {\sf Push-Pull} protocol can be effective even on poorly connected graphs.

It is worth mentioning that bounds for the number of rounds to inform \emph{almost all} vertices have already appeared in the literature, see for instance~\cite{DFF12,FPS12}.
In particular, for power-law Chung-Lu graphs with exponent in $(2,3)$, it is shown in~\cite{FPS12} that whp after $\Oh (\log \log n)$ rounds the rumor spreads in $n-o(n)$ vertices, but to inform \emph{all} vertices {of the giant component} $\Theta(\log n)$ rounds are needed.
This result also shows a great difference between the two cases, however in both cases the required time is quite small.

A closely related class of graphs is the class of \emph{random $k$-Apollonian networks},
introduced by Zhang, Comellas, Fertin, and Rong~\cite{high_RANs}.
Their construction is very similar to that of random $k$-trees, with just one difference:
if a $k$-clique is chosen in a certain round, it will never be chosen again.
It is known that whp random $k$-Apollonian networks exhibit a power law degree distribution and large clustering coefficient~\cite{RANs_powerlaw,define_RANs}
and have logarithmic diameter~\cite{CF13}.
Our third main contribution is the following theorem.

\begin{theorem}
\label{thm:apollonian}
{Let $k \ge 3$ be constant and let $f(n)=o(\log \log n)$ be an arbitrary function going to infinity with $n$.
Assume that initially a random vertex of an $(n+k)$-vertex random $k$-Apollonian network knows a rumor.
Then, with high probability after 
$$\Oh\left( (\log n) ^{(k^2-3)/(k-1)^2} \cdot \log \log n \cdot f(n)^{2/k} \right)$$ rounds of the {\sf Push-Pull} protocol, at least $n-o(n)$ vertices will know the rumor.}
\end{theorem}

{The proof of Theorem~\ref{thm:apollonian} is along the lines of that of Theorem~\ref{thm:mainrandomktree}, although there are several differences.
Note that we have ${(k^2-3)/(k-1)^2} < 1 + 2/k$,
so our upper bound for random $k$-Apollonian networks is slightly stronger than that for random $k$-trees.}

Unfortunately, our technique for proving Theorem~\ref{thm:lowerbound} does not extend to random $k$-Apollonian networks,
although we believe that whp we need a polynomial number of rounds to inform all vertices in a random $k$-Apollonian network as well.
We leave this as a conjecture.

For the rest of the paper, $k$ is a constant greater than 1, and the asymptotics are for $n$ going to infinity.
Several times in our proofs we use urn models to analyze the vertices' degrees and the number of vertices in certain parts of a random $k$-tree.
We also use a result on the height of random recursive trees to conclude that random $k$-trees have logarithmic diameter.
The connections with urn models are built in Section~\ref{sec:urns}.
In Section~\ref{expansion} we study basic properties of random $k$-trees, demonstrating their similarities with real-world graphs.
Theorems~\ref{thm:mainrandomktree},~\ref{thm:lowerbound}, and~\ref{thm:apollonian}
are proved in Sections~\ref{proofupper},~\ref{sec:lowerbound},
and~\ref{k-apol}, respectively.

\section{{Connections with} urn theory}
\label{sec:urns}
We will need some definitions and results from urn theory (see~\cite{urn_models_application} for a general introduction).
After reviewing these, we build some connections with random $k$-trees that will be used throughout.

\begin{definition}[P\'{o}lya-Eggenberger urn]
Start with $W_0$ white and $B_0$ black balls in an urn.
In every step a ball is drawn from the urn uniformly at random,
the ball is returned to the urn, and $s$ balls of the same color are added to the urn.
Let $\polya(W_0, B_0, s, n)$ denote the number of white balls right after $n$ draws.
\end{definition}

\begin{proposition}
\label{pro:polyasimple}
Let $X = \polya(a,b,k,n)$ and $w=a+b$. Then 
$$
\Ex{X^2} = \left(a + \frac{a}{w}\: kn\right)^2+\frac{a b k^2 n (kn + w)}{{w}^2 (w+k)}\:
$$
and for any $c \geq (a + b)/k$ we have
$$\Pr{X={a}}\leq\left( \frac{c}{c+n} \right)^{a/k}\:.$$
\end{proposition}

\begin{proof}
The first statement follows from the following well known formulae for the expected value and the variance of $X$  (see~\cite[Corollary~5.1.1]{urns_mahmoud} for instance):
\begin{align*}
\Ex{X} = a + \frac{a}{w}\: kn \:,\quad
\Var{X} = \frac{a b k^2 n (kn + w)}{{w}^2 (w+k)}\:.
\end{align*}
For the second statement, we have
\begin{align*}
\Pr{X=a} & = \frac{b}{a+b} \cdot \frac{b+k}{a+b+k} \cdot \cdots \frac{b+(n-1)k}{a+b+(n-1)k} \\
& = \prod_{i=0}^{n-1} \left(1 - \frac{a}{a+b+ik}\right) \\
& \le \prod_{i=0}^{n-1} \left(1 - \frac{a}{ck+ik}\right) \\
& \le  \exp\left(- \sum_{i=0}^{n-1} \frac{a}{ck+ik}\right) \\
& =  \left\{\exp\left( \sum_{i=0}^{n-1} \frac{1}{c+i}\right)\right\}^{-a/k}  \le  \left\{\exp\left( \int_{x=c}^{c+n} \frac{\mathrm{d}x}{x}\right)\right\}^{-a/k} 
= \left( \frac{c}{c+n} \right)^{a/k} \:.\qedhere
\end{align*}
\end{proof}

\begin{definition}[Generalized P\'{o}lya-Eggenberger urn]
Let $\alpha,\beta,\gamma,\delta$ be nonnegative integers.
We start with $W_0$ white and $B_0$ black balls in an urn.
In every step a ball is drawn from the urn uniformly at random and returned to the urn.
Additionally, if the ball is white, then $\delta$ white  balls and $\gamma$  black balls are returned to the urn;
otherwise, i.e.\ if the ball is black, then $\beta$ white  balls and $\alpha$  black balls are returned to the urn.
Let $\polya\left(W_0, B_0, \begin{bmatrix}\alpha&\beta\\ \gamma&\delta\end{bmatrix}, n\right)$ denote the number of white balls right after $n$ draws.
\end{definition}
Note that P\'{o}lya-Eggenberger urns correspond to the matrix $\begin{bmatrix}s&0\\0&s\end{bmatrix}$.
The following proposition follows from known results.
\begin{proposition}
\label{pro:moments}
Let $X = \polya\left(W_0, B_0, \begin{bmatrix}\alpha&0\\ \gamma&\delta\end{bmatrix}, n\right)$ and let $r$ be a positive integer.
If $\gamma,\delta>0$, $\alpha=\gamma+\delta$, and $r\delta \ge \alpha$, then we have
$$\Ex{X^r} \le
 \left( \frac{\alpha n}{W_0+B_0} \right)^{r\delta/\alpha}\
 \prod_{i=0}^{r-1} \left( {W_0} + i \delta \right) + \Oh\left(n^{(r-1)\delta/\alpha}\right)
 \:.
$$
\end{proposition}

\begin{proof}
By~\cite[Proposition~15]{exactly_solvable}
we have 
$$
\Ex{X^r} = n^{r\delta/\alpha} \delta^r
\frac{\Gamma(W_0/\delta + r)\Gamma((W_0+B_0)/\alpha)}{\Gamma(W_0/\delta)\Gamma((W_0+B_0+r\delta)/\alpha)} + \Oh\left(n^{(r-1)\delta/\alpha}\right)\:,
$$
Note that
$$\frac{\Gamma(W_0/\delta + r)}{\Gamma(W_0/\delta)}
=  \prod_{i=0}^{r-1} \left( i + {W_0}/\delta \right).
$$
Finally, the inequality
$$
\frac{\Gamma((W_0+B_0+r\delta)/\alpha)} {\Gamma((W_0+B_0)/\alpha)} \ge \left((W_0+B_0)/\alpha\right)^{r\delta/\alpha}
$$
follows from $r\delta \ge \alpha$ and the following inequality (see, e.g.,~\cite[inequality (2.2)]{gamma_inequalities})
$$\frac{\mathop{\Gamma\/}\nolimits\!\left(x+1\right)}{\mathop{\Gamma\/
}\nolimits\!\left(x+s\right)} \ge x^{{1-s}} \qquad \forall x>0, s\in[0,1] \:.\qedhere
$$
\end{proof}

\begin{proposition}
\label{pro:degs}
Suppose that in $G(j)$
vertex $x$ has $A>0$ neighbors, and is contained in $B$ many $k$-cliques.
Conditional on this, the degree of $x$  in $G(n+j)$ is distributed as
$$A + \left( \polya\left(B,kj+1-B,\begin{bmatrix}k&0\\ 1&k-1 \end{bmatrix},n\right) - B\right)\Big/(k-1) \:.$$
\end{proposition}

\begin{proof}
We claim that
the total number of $k$-cliques containing $x$ in $G(n+j)$ is distributed as 
$\polya\left(B,kj+1-B,\begin{bmatrix}k&0\\ 1&k-1 \end{bmatrix},n\right)$.
At the end of round $j$, there are $B$ many $k$-cliques containing $x$,
and $kj+1-B$ many $k$-cliques not containing $x$.
In each subsequent round $j+1,\dots,j+n$,
a random $k$-clique is chosen and $k$ new $k$-cliques are created.
If the chosen $k$-clique contains $x$,
then $k-1$ new $k$-cliques containing $x$ are created, and $1$ new $k$-clique not containing $x$ is created.
Otherwise, i.e.\ if the chosen $k$-clique does not contain $x$,
then no new $k$-cliques containing $x$ is created, and $k$ new $k$-cliques not containing $x$ are created, and the claim follows.

Hence the number of $k$-cliques that are created in rounds $j+1,\dots,j+n$ and contain $x$ is
{distributed as} $\polya\left(B,kj+1-B,\begin{bmatrix}k&0\\ 1&k-1 \end{bmatrix},n\right) - B$,
and the proof follows by noting that every new neighbor of $x$ creates $k-1$ new $k$-cliques containing $x$.
\end{proof}
Combining Propositions~\ref{pro:moments} and~\ref{pro:degs} we obtain the following lemma.
\begin{lemma}
\label{lem:largedegree}
Let $1\le j \le n$ and let $q$ be a positive integer.
Let $x$ denote the vertex born in round $j$.
Conditional on any $G(j)$, the probability that $x$
has degree greater than $k + q (n/j)^{(k-1)/k}$ in $G(n)$ is $\Oh\left(q\sqrt{q}\exp(-q)\right)$.
\end{lemma}

\begin{proof}
Let $X=\polya\left(k, kj-k+1, \begin{bmatrix}k&0\\ 1&k-1 \end{bmatrix}, n-j\right)$.
By Proposition~\ref{pro:degs}, $\deg(x)$ is  distributed as $k + \left(X-k\right)/(k-1)$.
By Proposition~\ref{pro:moments},
$$\Ex{X^q} \le
(1+o(1)) \left( \frac{k (n-j)}{kj+1} \right)^{\frac{q(k-1)}{k}}\
 \prod_{i=0}^{q-1} \left( k + i (k-1) \right)
\le \left( \frac{n}{j} \right)^{\frac{q(k-1)}{k}} (k-1)^{q} (q+1)!
 \:.
$$
Thus,
\begin{align*}
\Pr{\deg(x) > k + q (n/j)^{(k-1)/k}} &=
\Pr{X-k > q (k-1)(n/j)^{(k-1)/k}} \\
& \le \frac{\Ex{X^q}}{\left(q (k-1)(n/j)^{(k-1)/k}\right)^q} \\
& \le  (q+1)!q^{-q} = \Oh\left(q\sqrt{q}\exp(-q)\right) \:. \qedhere
\end{align*}
\end{proof}

\section{Basic properties of random $k$-trees}
\label{expansion}
In this section we prove that random $k$-trees exhibit two important properties observed in real-world networks:
low diameter and large clustering coefficient.
We also prove that random $k$-trees do not expand well,
confirming our claim in the introduction that random $k$-trees are poorly connected graphs and thus existing techniques do not apply.
Let $G(0),G(1),\dots$ be defined as in Definition~\ref{def_random_k_tree}.

\begin{definition}
The
\emph{clustering coefficient} of a graph $G$, written $cc(G)$, is defined as
$$cc(G) = \frac{1}{|V(G)|}\sum_{u\in V(G)}\frac{|\langle N(u)\rangle|}{\binom{\deg(u)}{2} } \:,$$
where $|\langle N(u)\rangle|$ denotes the number of edges $xy$ such that both $x$ and $y$ are neighbors of $u$.
\end{definition}

\begin{proposition}
\label{pro:clustering}
For every positive integer $n$, {deterministically,} the clustering coefficient of $G(n)$ is at least 1/2.
\end{proposition}

\begin{proof}
Let $u$ be a vertex of $G=G(n)$.
It is not hard to check that $|\langle N(u)\rangle| = (k-1) (\deg(u)-k/2)$,
and since $\deg(u) \ge k$ we get
$$\frac{|\langle N(u)\rangle|}{\binom{\deg(u)}{2} } \ge \frac{k}{\deg(u)} \:.$$
Using the Cauchy-Schwarz inequality we get
\begin{equation*}
cc(G) \ge \frac{1}{|V(G)|}\sum_{u\in V(G)}\frac{k}{\deg(u)}
\ge \frac{k}{n+k} \cdot \frac{(n+k)^2}{2|E(G)|} \ge \frac 1 2 \:.\qedhere
\end{equation*}
\end{proof}

For proving random $k$-trees have logarithmic diameter we will need a known result about random $d$-ary recursive trees.

\begin{definition}[Random $d$-ary recursive tree]
\label{def:randomrecursive}
Let $d$ be a positive integer.
Build a sequence $T(0)$, $T(1),$ $\dots$ of rooted random trees as follows.
The tree $T(0)$ has just one vertex, the root.
For each $1\leq t\leq n$, $T(t)$ is obtained from $T(t-1)$ as follows:
a leaf of $T(t-1)$ is chosen uniformly at random
and gives birth to $d$ new children.
The tree $T(n)$ is called a random $d$-ary recursive tree on $dn+1$ vertices.
\end{definition}

\begin{theorem}
[\cite{randomtrees},Theorem~6.47]
\label{pro:drmota}
Let $\alpha$ be the unique solution in $(d,\infty)$ of 
$$\alpha (d-1) \log \left( \frac{de}{\alpha(d-1)}\right) = 1 \:.$$
Let $H_n$ denote the height of a random $d$-ary recursive tree on $dn+1$ vertices.
There exists a constant $c>0$ such that for any $\eta$,
$$\Pr{H_n > \alpha \log n + \eta} = \Oh(e^{-c\eta}) \:.$$
\end{theorem}

{The following proposition implies that
with high probability the diameter of $G(n)$ is $\Oh (\log n)$.

\begin{proposition}
\label{pro:diameter}
Whp $G(n)$ has the following property:
let $u_h u_{h-1} \cdots u_0$ be an arbitrary path
such that $u_i$ is born later than $u_{i-1}$
for all $i$;
then $h=\Oh(\log n)$.
\end{proposition}}

\begin{proof}
We inductively define a notion of \emph{draft} for vertices and $k$-cliques of $G(n)$.
The draft of the vertices of $G(0)$ as well as the $k$-clique they form equals 0.
The draft of every $k$-clique equals the maximum draft of its vertices.
Whenever a new vertex is born and is joined to a $k$-clique, the draft of the vertex
equals the draft of the $k$-clique plus one.
It is easy to see that if $xy \in E(G(n))$ and $x$ is born later than $y$, then $\operatorname{draft}(x) \ge \operatorname{draft}(y)+1$.
In particular, if $x$ is a vertex of $G(n)$ 
and there is a path $x = x_h, x_{h-1}, \dots, x_1, x_0$ in $G(n)$ such that
$x_j$ is born later than $x_{j-1}$ for each $j$,
then $\operatorname{draft}(x) \ge h$.
Hence we just need to show that with high probability the draft of each $k$-clique of $G(n)$ is $\Oh(\log n)$.

We define an auxiliary tree whose nodes correspond to the $k$-cliques of $G(n)$,
{in such a way that the depth of each node in this tree equals the  draft of its corresponding $k$-clique.}
Start with a single node corresponding to $G(0)$.
Whenever a new vertex $x$ is born and is joined to a $k$-clique $C$,
$k$ new $k$-cliques are created.
In the auxiliary tree, add these to the set of children of $C$.
{So, the auxiliary tree evolves as follows:
in every round a node is chosen uniformly at random and gives birth to $k$ new children. Hence,}
the height of the auxiliary tree {after $n$ rounds} is stochastically smaller than that of a random $k$-ary recursive tree on $1+kn$ nodes,
whose height is $\Oh(\log n)$ whp
by Theorem~\ref{pro:drmota}.
\end{proof}

\begin{definition}
The \emph{vertex expansion}
of a graph $G$
(also known as the \emph{vertex isoperimetric number} of $G$), written $\alpha(G)$,
is defined as
\[
\alpha(G)=\min\left\{\frac{|\partial S|}{|S|} :  S\subseteq V(G), 0<|S|\leq |V(G)|/2  \right\} \:,
\]
 where $\partial S$ denotes the set of vertices in $V(G)\setminus S$ that have a neighbor in $S$.
\end{definition}

\begin{definition}
The \emph{conductance} of a graph $G$, written  $\Phi(G)$, is defined as
\[
\Phi(G)=\min\left\{\frac{e(S, V(G)\setminus S)}{{\tt vol}(S)} : S\subseteq V(G),
0 < {\tt vol}(S) \le {\tt vol}(V(G)) / 2 \right\} \:,
\]
where ${ e}(S, V(G)\setminus S)$ denotes the number of edges between $S$ and $V(G)\setminus S$,
and ${\tt vol}(S)=\sum_{u\in S}\deg(u)$ for every $S\subseteq V(G)$.
\end{definition}

\begin{proposition}
Deterministically $G(n)$ has
vertex expansion $\Oh\left(k/n\right)$, and whp its
conductance is $\Oh\left(\log n \cdot n^{-1/k}\right)$.
\end{proposition}
\begin{proof}
Let $G=G(n)$.
Since $G$ has treewidth $k$, by~\cite[Lemma~5.3.1]{Kl94}  there exists a partition $(A,B,C)$ of $V(G)$ such that
\begin{enumerate}
\item $|C|=k+1$,
\item $(n-1)/3 \le |A| \le 2(n-1)/3$ and
$(n-1)/3 \le |B| \le 2(n-1)/3$, and
\item there is no edge between $A$ and $B$.
\end{enumerate}
At least one of $A$ and $B$, say $A$, has size less than  $(n+k)/2$.
Then
$$\alpha(G) \le \frac{|\partial A|}{|A|} \le \frac{k+1}{(n-1)/3} = \Oh (k/n) \:.$$

At least one of $A$ and $B$, say $B$, has volume less than  ${\tt vol} (G)/2$.
Then since all vertices in $G$ have degrees at least $k$,
$$\Phi(G) \le \frac{e(B, A\cup C)}{{\tt vol}(B)} \le \frac{e(B,C)}{k|B|} \le \frac{(k+1)\Delta(G)}{k(n-1)/3} =\Oh(\Delta(G)/n)\:.$$
Hence to prove
$\Phi(G) = \Oh\left(\log n \cdot n^{-1/k}\right)$
it suffices to show that with high probability we have
$$\Delta(G) \le k + (2 \log n) n^{1-1/k} \:.$$
Let $q = \left \lfloor 2 \log n \right \rfloor$ and let $x$ be a vertex born in one of the rounds $1,2,\dots,n$.
By Lemma~\ref{lem:largedegree},
$$\Pr{\deg(x) > k + q n^{1-1/k}} = \Oh (q\sqrt{q} \exp(-q)) = o(1/n) \:.$$
An argument similar to the proof of Lemma~\ref{lem:largedegree} shows that the probability that a vertex in $G(0)$
has degree greater than $k + q n^{1-1/k}$ is $o(1/n)$ as well.
A union bound over all vertices shows that with high probability we have
$\Delta(G) \le k + (2 \log n) n^{1-1/k}$, as required.
\end{proof}

\section{Proof of Theorem~\ref{thm:mainrandomktree}}
\label{proofupper}
Once we have the following lemma, our problem reduces to proving a structural result for random $k$-trees.

\begin{lemma}
\label{lem:feige}
{Let $\chi$ and $\tau$ be fixed positive integers.}
Let $G$ be an $n$-vertex graph and let $\Sigma\subseteq V(G)$ with $|\Sigma|=n-o(n)$ be such that for every pair of vertices
$u,v\in \Sigma$ there exists a $(u,v)$-path $u u_1 u_2 \cdots u_{l-1} v$
such that $l\le \chi$ and for every $0\le i \le l-1$ we have $\min\{\deg(u_i),\deg(u_{i+1})\} \le \tau$
(where we define $u_0=u$ and $u_l=v$).
If a random vertex in $G$ knows a rumor, then whp after
$6 \tau (\chi + \log n)$
rounds of the {\sf Push-Pull} protocol, at least $n-o(n)$ vertices will know the rumor.
\end{lemma}

\begin{proof}
{The proof is along the lines of that of~\cite[Theorem~2.2]{FPRU90}.}
We show that given any $u,v \in \Sigma$,
if $u$ knows the rumor then with probability at least $1 - o\left(n^{-2}\right)$ after $ 6 \tau (\chi + \log n)$ rounds $v$ will know the rumor.
The lemma follows by using the union bound and noting that a random vertex lies in $\Sigma$ with high probability.
Consider the $(u,v)$-path $u u_1 u_2 \cdots u_{l-1} v$ promised by the hypothesis.
We bound from below the probability that the rumor is passed through this path.

For every $0\leq i\leq l-1$, the number of rounds taken for the rumor to pass from $u_i$ to $u_{i+1}$
is a geometric random variable with success probability at least $1/\tau$
(if $\deg(u_i) \le \tau$, this is the number of rounds needed for $u_i$ to push the rumor along the edge,
and if $\deg(u_{i+1}) \le \tau$, this is the number of rounds needed for $u_{i+1}$ to pull the rumor along the edge).
The random variables corresponding to distinct edges are mutually independent.
Hence the probability that the rumor is not passed in $6 \tau (\chi + \log n)$
 rounds is at most the probability that
the number of heads in a sequence of $6 \tau (\chi + \log n)$
 independent biased coin flips, each having probability $1/\tau$ of being heads,
is less than $l$.
Let $X$ denote the number of heads in such a sequence.
Then using the Chernoff bound (see, e.g., \cite[Theorem 4.2]{rand_algs}) and noting that $\Ex{X} = 6 (\chi + \log n)$ we get
\begin{align*}
\Pr{X < l} \le
\Pr{ X \leq \Ex{X} / 6} \leq \exp(-(5/6)^2 \Ex{X} / 2)
\leq \exp(-(5/6)^2 (6 \log n) / 2)  \:,
\end{align*}
which is $o\left(n^{-2}\right)$, as required.
\end{proof}

Let $f(n)=o(\log \log n)$ be an arbitrary function going to infinity with $n$, and let
$$m = \left\lceil 
\frac{n}{f(n)^{{3}/{(k-1)}} (\log n)^{{2}/{(k-1)}}} \right \rceil\:.$$
Also let $q = \lceil 4 \log \log n \rceil$ and let
\begin{equation}
\label{tau_def}
\tau = 2k + q (n/m)^{1-1/k} \:.
\end{equation}
By Proposition~\ref{pro:diameter}, whp a random $k$-tree on $n+k$ vertices has diameter $\Oh(\log n)$.
Theorem~\ref{thm:mainrandomktree} thus follows from Lemma~\ref{lem:feige} and the following structural result,
which we prove in the rest of this section.

\begin{lemma}
\label{thm:structural}
Let $G$ be an $(n+k)$-vertex random $k$-tree.
Whp there exists  $\Sigma \subseteq V(G)$ satisfying the conditions of Lemma~\ref{lem:feige}
with $\tau$ defined in (\ref{tau_def}) and $\chi = \Oh(\log n + \operatorname{diam}(G))$.
\end{lemma}

For the rest of this section, $G$ is an $(n+k)$-vertex random $k$-tree.
Recall from Definition~\ref{def_random_k_tree} that $G=G(n)$, where
$G(0),G(1),\dots$ is the random $k$-tree process.
Consider the graph $G(m)$, which has $k+m$ vertices and $mk+1$ many $k$-cliques.
For an edge $e$ of $G(m)$, let $N(e)$ denote the number of $k$-cliques of $G(m)$ containing  $e$.
We define a spanning forest $F$ of $G(m)$ as follows:
for every $1\le t\le m$,
if the vertex $x$ born in round $t$ is joined to the $k$-clique $C$,
then in $F$, $x$ is joined to a vertex $u \in V(C)$ such that
$$N(x u) = \max_{v\in V(C)}N(x v) \:.$$

Note that $F$ has $k$ trees and the $k$ vertices of $G(0)$ lie in distinct trees.
Think of these trees as rooted at these vertices.
The tree obtained from $F$ by merging these $k$ vertices is the `highway system' described in the sketch of the proof of Theorem~\ref{thm:mainrandomktree}.
Informally speaking, the proof has three parts:
first, we show that this tree has a small height (Lemma~\ref{lem:log});
second, we show that each edge in this tree quickly exchanges the rumor with a reasonably large probability (Lemma~\ref{FastEdges});
and finally we show that almost all vertices in $G-G(m)$ have quick access to and from $F$ (Lemma~\ref{lem:nonnice}).

Let $\mathsf{LOG}$ denote the event `each tree in $F$ has height $\Oh(\log n)$.'
{The following lemma is an immediate corollary of
Proposition~\ref{pro:diameter}.}

\begin{lemma}
\label{lem:log}
With high probability $\mathsf{LOG}$ happens.
\end{lemma}

We prove Lemma~\ref{thm:structural} conditional on the event $\mathsf{LOG}$.
In fact, we prove it for any $G(m)$ that satisfies $\mathsf{LOG}$.
Let $G_1$ be an arbitrary instance of $G(m)$ that satisfies $\mathsf{LOG}$.
{So, $G_1$ and $F$ are fixed in the following, and}
all randomness refers to rounds $m+1,\dots,n$.
The following deterministic lemma
will be used in the proof of Lemma~\ref{FastEdges}.

\begin{lemma}
\label{lem:counting_new}
Assume that  $x y \in E(F)$ and  $x$ is born later than $y$.
If the degree of $x$ in $G_1$ is greater than $2k-2$, then $N(xy) \ge (k^2-k)/2$.
\end{lemma}

\begin{proof}
Assume that $x$ is joined to $u_1,\dots,u_k$ when it is born, and that $v_1,v_2,\dots,v_{k-1},\dots$ are the neighbors of $x$ that are born later than $x$,
in the order of birth.
Let $\Psi$ denote the number of pairs $(u_j,C)$, where $1\le j \le k$, and
 $C$ is a $k$-clique in $G_1$ containing the edge  $x u_j$.
Consider the round in which vertex $x$ is born and is joined to $u_1,\dots,u_k$.
For every $j\in\{1,\dots,k\}$, the vertex $u_j$ is contained in $k-1$ new $k$-cliques,
so in this round $\Psi$ increases by $k(k-1)$.
For each $i\in\{1,\dots, k-1\}$, consider the round in which vertex $v_i$ is born.
This vertex is joined to $x$ and $k-1$ neighbors of $x$.
At this round $x$ has neighbor set $\{u_1,\dots,u_k,v_1,\dots,v_{i-1}\}$.
Thus at least $k-i$ of the $u_j$'s are joined to $v_i$ in this round.
Each vertex $u_j$ that is joined to $v_i$ in this round is contained in $k-2$ new $k$-cliques that contain $x$ as well,
so in this round $\Psi$ increases by at least $(k-i)(k-2)$.
Consequently, we have
$$\Psi \ge k(k-1) + \sum_{i=1}^{k-1} (k-i)(k-2) = k^2(k-1)/2 \:.$$
By the pigeonhole principle, there exists some $\ell\in\{1,\dots, k\}$ such that
the edge $x u_{\ell}$ is contained in at least $(k^2-k)/2$ many $k$-cliques, and this completes the proof.
\end{proof}

A vertex of $G$ is called \emph{modern} if it is born later than the end of round $m$,
and is called \emph{traditional} otherwise.
In other words, vertices of $G_1$ are traditional and vertices of $G-G_1$ are modern.
We say edge $uv\in E(G)$ is \emph{fast} if at least one of the following is true:
$\deg(u)\le \tau$, or $\deg(v)\le \tau$, or $u$ and $v$ have a common neighbor $w$ with $\deg(w)\le \tau$.
For an edge $uv \in E(F)$, let $p_S(uv)$ denote the probability that $uv$ is not fast,
and let $p_S$ denote the maximum of $p_S$ over all edges of $F$.

\begin{lemma}\label{FastEdges}
We have $p_S  = o ( 1 / (f(n) \log n))$.
\end{lemma}

\begin{proof}
Let $xy \in E(F)$ be arbitrary.
By symmetry we may assume that $x$ is born later than $y$.
By Lemma~\ref{lem:counting_new}, at least one of the
following is true: vertex $x$ has less than $2k-1$ neighbors in $G_1$,
or $N(xy) \ge (k^2-k)/2$.
So we may consider two cases.
  \begin{itemize}
\item Case 1: vertex $x$ has less than $2k-1$ neighbors in $G_1$.
In this case vertex $x$ lies in at most $k^2-2k+2$ many $k$-cliques of $G_1$.
Assume that $x$ has $A$ neighbors in $G_1$ and lies in $B$ many $k$-cliques in $G_1$.
Let
$$X = \polya\left(B,km+1-B,\begin{bmatrix}k&0\\ 1&k-1 \end{bmatrix},n-m\right)\:.$$
Then by Proposition~\ref{pro:degs} the degree of $x$ is distributed as
$A + \left( X - B\right)/(k-1)$.
By Proposition~\ref{pro:moments},
\begin{align*}
\Ex{X^q} & \le
 (1+o(1))
 \left( \frac{k (n-m)}{km+1} \right)^{\frac{q(k-1)}{k}}
 \prod_{i=0}^{q-1} \left( B + i (k-1) \right) \\
& \le
 (1+o(1))\left( \frac{n}{m} \right)^{\frac{q(k-1)}{k}}
 (k-1)^q \prod_{i=0}^{q-1} \left( k + i \right) \\
 &\le (k-1)^q (k+q)!  \left( \frac{n}{m} \right)^{\frac{q(k-1)}{k}} \:,
\end{align*}
where we have used $B \le k(k-1)$ for the second inequality.
Therefore,
\begingroup
\addtolength{\jot}{1em}
\begin{align*}
& \Pr{\deg(x) > 2k + q (n/m)^{\frac{k-1}{k}}}  \le \Pr{X \ge (k-1)q (n/m)^{\frac{k-1}{k}}} \\
& \le \frac{\Ex{X^q}}{(k-1)^qq^q (n/m)^{\frac{q(k-1)}{k}}} = \Oh \left( \frac{(k+q)^{k+q} \sqrt{q}}{q^q \exp(k+q)} \right) = o \left( \frac{1}{f(n) \log n}\right) \:.
\end{align*}
\endgroup

\item Case 2: $N(xy) \ge (k^2-k)/2$.
In this case we bound from below the probability that there exists a modern vertex $w$
that is adjacent to $x$ and $y$ and has degree at most  $\tau$.
We first bound from above the probability that $x$ and $y$ have no modern common neighbors.
For this to happen, none of the $k$-cliques containing $x$ and $y$ must be chosen in rounds $m+1,\dots,n$.
This probability equals $\Pr{\polya(N(xy),mk+1-N(xy),k,n-m)={N(xy)}}$.
Since $N(xy) \ge (k^2-k)/2$, by Proposition~\ref{pro:polyasimple} we have
\begin{equation*}
\Pr{\polya(N(xy),mk+1-N(xy),k,n-m)=N(xy)} \le \left( \frac{m+1}{n+1}\right)^{\frac{k-1  }{2 }}\:,
\end{equation*}
which is $o \left( 1/(f(n) \log n)\right)$.

Now, assume that $x$ and $y$ have a modern common neighbor $w$.
If there are multiple such vertices, choose the one that is born first.
Since $w$ appears later than round $m$, by
Lemma~\ref{lem:largedegree},
$$\Pr{\deg(w) > k + q (n/m)^{(k-1)/k}} =\Oh\left( q\sqrt{q}\exp(-q) \right) =o \left( \frac{1}{f(n) \log n}\right) \:.\qedhere$$
\end{itemize}
\end{proof}

Enumerate the $k$-cliques of $G_1$ as $C_1,\dots,C_{mk+1}$.
Then choose $r_1\in C_1,\dots,$ $r_{mk+1}\in C_{mk+1}$ arbitrarily, and call them the \emph{representative vertices}.
Starting from $G_1$, when modern vertices are born in rounds $m+1,\dots,n$ until $G$ is formed,
every clique $C_i$ `grows' to a random $k$-tree with a random number of vertices, which is a subgraph of $G$.
Enumerate these subgraphs as $H_1,\dots,H_{mk+1}$, and call them the \emph{pieces}.
More formally, $H_1,\dots,H_{mk+1}$ are induced subgraphs of $G$ such that
a vertex $v$ is in $V(H_j)$ if and only if every path connecting $v$ to a traditional vertex intersects $V(C_j)$.
In particular, $V(C_j)\subseteq V(H_j)$ for all $j\in\{1,\dots,mk+1\}$.
Note that the $H_j$'s may intersect, as a traditional vertex may lie in more than one $C_j$,
however every modern vertex lies in a unique piece.

A traditional vertex is called \emph{nice} if it is connected to some vertex in $G(0)$ via a path of fast edges.
Since $F$ has height $\Oh(\log n)$ and each edge of $F$ is fast with probability at least $1-p_S$,
the probability that a given traditional vertex is not nice is $\Oh(p_S \log n)$ by the union bound.
A piece $H_j$ is called \emph{nice} if all its modern vertices have degrees at most $\tau$,
and the vertex $r_j$ is nice.
A modern vertex is called \emph{nice} if it lies in a nice piece.
A vertex/piece is called \emph{bad} if it is not nice.

\begin{lemma}
\label{lem:nonnice}
The expected number of bad vertices is $o(n)$.
\end{lemma}

\begin{proof}
The total number of traditional vertices is $k+m=o(n)$ so we may just ignore them in the calculations below.
Let $\eta = n f(n) / m = o (\log^3 n)$.
Say piece $H_j$ is \emph{sparse} if $|V(H_j)| \le \eta + k$.
We first bound the expected number of modern vertices in non-sparse pieces.
Observe that the number of modern vertices in a given piece is distributed as $X = (\polya(1,km,k,n-m)-1)/k$.
Using Proposition~\ref{pro:polyasimple} we get
$\Ex{X^2} \le {2kn^2}/{m^2}$.
By the second moment method, for every $t>0$ we have
$$\Pr{X \ge t} \le \frac{\Ex{X^2}}{t^2} \le \frac{2kn^2}{m^2t^2} \:.$$
The expected number of modern vertices in non-sparse pieces is thus at most
\begin{align*}
(km+1)
\sum_{i=0}^{\infty}
(2^{i+1} \eta)   \Pr {2^i \eta < X \le 2^{i+1} \eta}
& \le
\sum_{i=0}^{\infty}
(2^{i+1} \eta) (km+1)  \frac{2kn^2}{m^2\eta^2 2^{2i}} \\
& \le
\Oh\left(\frac{n^2}{m\eta} \right)
\sum_{i=0}^{\infty} 2^{-i} = \Oh\left(\frac{n^2}{m\eta} \right)\:,
\end{align*}
which is $o(n)$.

We now bound the expected number of modern vertices in sparse bad pieces.
For bounding this from above we find an upper bound for the expected number of bad pieces, and multiply by $\eta$.
A piece $H_j$ can be bad in two ways:

(1) the representative vertex $r_j$ is bad: the probability of this is $\Oh\left(p_S \log n\right)$.
Therefore, the expected number of pieces that are bad due to this reason is $\Oh\left(mk p_S \log n\right)$,
which is $o(n/\eta)$ by Lemma~\ref{FastEdges}.

(2) there exists a modern vertex in $H_j$ with degree greater than $\tau$:
the probability that a given modern vertex has degree greater than $\tau$
is $\Oh\left(q\sqrt{q}\exp(-q)\right)$ by Lemma~\ref{lem:largedegree}.
So the average number of modern vertices with degree greater than $\tau$ is  $\Oh\left(n q\sqrt{q}\exp(-q)\right)$.
Since every modern vertex lies in a unique piece,
the expected number of pieces that are bad because of this reason is bounded by  $\Oh\left(n q\sqrt{q}\exp(-q)\right)=o(n/\log^3n)$.

So the expected number of bad pieces is $o(n/\eta + n/\log^3n)$, and
the expected number of modern vertices in sparse bad pieces is $o(n + \eta n/\log^3n) = o(n)$.
\end{proof}

We now  prove Lemma~\ref{thm:structural},
which concludes the proof of Theorem~\ref{thm:mainrandomktree}.

\begin{proof}[Proof of Lemma~\ref{thm:structural}]
Let $\Sigma$ denote the set of nice modern vertices.
By Lemma~\ref{lem:nonnice} and using Markov's inequality, we have $|\Sigma|=n-o(n)$ whp.
Let $\{a_1,\dots,a_k\}$ denote the vertex set of $G(0)$.
Using an argument similar to the proof of Lemma~\ref{FastEdges}, it can be proved that
given $1\le i < j \le k$, the probability that edge $a_i a_j$ is not fast is $o(1)$.
Since the total number of such edges is a constant, whp all such edges are fast.
Let $u$ and $v$ be nice modern vertices,
and let $r_u$ and $r_v$ be the representative vertices of the pieces containing them, respectively.
Since the piece containing $u$ is nice, there exists a $(u,r_u)$-path whose vertices except possibly $r_u$
all have degrees at most $\tau$.
The length of this path is at most $\operatorname{diam} (G)$.
Since $r_u$ is nice,  for some $1\le i \le n$ there exists an $(r_u,a_i)$-path in $F$ consisting of fast edges.
Appending these paths gives a $(u,a_i)$-path with length at most $\diam(G) + \Oh(\log n)$
such that for every pair of consecutive vertices in this path, one of them has degree at most $\tau$.
Similarly, for some $1\le j \le n$ there exists a $(v,a_j)$-path of length $\Oh(\log n + \operatorname{diam}(G))$,
such that one of every pair of consecutive vertices in this path has degree at most $\tau$.
Since the edge $a_i a_j$ is fast whp, we can build a $(u,v)$-path of length $\Oh(\log n + \operatorname{diam}(G))$
of the type required by Lemma~\ref{lem:feige}, and this completes the proof.
\end{proof}

\section{Proof of Theorem~\ref{thm:lowerbound}}
\label{sec:lowerbound}

\begin{definition}[$s$-barrier]
A pair $\{C_1,C_2\}$ of disjoint $k$-cliques in a connected graph is an \emph{$s$-barrier} if
(i) the set of edges between $C_1$ and $C_2$ is a cut-set, i.e.\ deleting them disconnects the graph, and
(ii) the degree of each vertex in $V(C_1)\cup V(C_2)$ is at least $s$.
\end{definition}

Observe that if $G$ has an $s$-barrier, then for any starting vertex, whp the \textsf{Push-Pull} protocol needs at least $\Omega(s)$ rounds to inform all vertices.

\begin{lemma}
\label{lem:barrier}
The graph $G(n)$ has an $\Omega(n^{1-1/k})$-barrier
with probability $\Omega(n^{1/k-k})$.
\end{lemma}

\begin{proof}
Let $u_1,\dots,u_k$ be the vertices of $G(0)$,
and let $v_1,\dots,v_k$ be the vertices of $G(k)-G(0)$ in the order of appearance.
We define two events.
Event A is that for every $1\le i \le k$, 
when $v_i$ appears, it attaches to 
$v_1,v_2,\dots,v_{i-1},u_i,u_{i+1},\dots,u_k$; and 
for each $1\le i,j \le k$, $u_i$ and $v_j$ have no common neighbor in $G(n)-G(k)$.
Event B is that all vertices of $G(k)$ have degree $\Omega(n^{(k-1)/k})$ in $G(n)$.
Note that if A and B both happen, then 
the pair $\{u_1u_2\dots u_k, v_1v_2\dots v_k\}$ is an $\Omega(n^{(k-1)/k})$-barrier in $G(n)$.
To prove the lemma we will show that $\Pr{A}=\Omega(n^{1/k-k})$ and
$\Pr{B | A} = \Omega(1)$.

For A to happen, first, the vertices $v_1,\dots,v_k$ must choose the specified $k$-cliques, which happens with constant probability.
Moreover, the vertices appearing after round $k$ must not choose any of the $k^2-1$ many $k$-cliques that contain both $u_i$'s and $v_j$'s.
Since $1-y\ge e^{-y-y^2}$ for every $y\in[0,1/4]$,
\begin{align*}
\Pr{A}&=\Omega(\Pr{\polya(k^2-1,2,k,n-k)={k^2-1}})\\
&=\Omega\left(\prod_{i=0}^{n-k-1} \left(\frac{2+ik }{k^2 + 1 + ik}\right)\right)\\
&\ge \Omega\left( \prod_{i=0}^{4k-1} \left(\frac{2+ik }{k^2 + 1 + ik}\right)  \prod_{i=4k}^{n-k-1} \left(1-\frac{k^2-1 }{ik}\right) \right)\\
&\ge \Omega\left(\exp\left(-\sum_{i=4k}^{n-k-1} \left\{\frac{k^2-1 }{ik}+\left(\frac{k^2-1 }{ik}\right)^2\right\}\right)\right)
\end{align*}
which is $\Omega(n^{1/k-k})$
since
$$
\sum_{i=4k}^{n-k-1} \frac{k^2-1}{ik}
\le (k-1/k) \log n + \Oh(1)\mathrm{\ and\ }
\sum_{i=4k}^{n-k-1} \left(\frac{k^2-1 }{ik}\right)^2=\Oh(1) \:.$$

Conditional on A and using an argument similar to that in the proof of Proposition~\ref{pro:degs}, the degree of each of the vertices $u_1,\dots,u_k,v_1,\dots,v_k$ in $G(n)$ is at least $k + (\polya(1,1,\begin{bmatrix}k&0\\ 1&k-1\end{bmatrix},n-k)-1)/(k-1)$.
By \cite[Proposition~16]{exactly_solvable},
there exists $\delta>0$ such that
$$\Pr{\polya(1,1,\begin{bmatrix}k&0\\ 1&k-1\end{bmatrix},n-k) < \delta n^{(k-1)/k} } < 1 / (2k+1)\:.$$
By the union bound, the probability that all vertices $u_1,\dots,u_k,v_1,\dots,v_k$ have degrees at least $\delta n ^{(k-1)/k} / (k-1)$ is at least $1/(2k+1)$,
hence $\Pr{B|A} \ge 1 / (2k+1) = \Omega(1)$.
\end{proof}

Let $f(n)=o(\log \log n)$ be any function going to infinity with $n$, and also let
$m = \left\lceil f(n) n ^{1- k/(k^2+k-1)} \right\rceil \:.$
(Note that the value of $m$ is different from that in Section~\ref{proofupper}, although its role is somewhat similar.)
Consider the random $k$-tree process up to round $m$.
Enumerate the $k$-cliques of $G(m)$ as $C_1,\dots,C_{mk+1}$.
Starting from $G(m)$, when new vertices are born in rounds $m+1,\dots,n$ until $G=G(n)$ is formed,
every clique $C_i$ `grows' to a random $k$-tree with a random number of vertices, which is a subgraph of $G$.
Enumerate these subgraphs as $H_1,\dots,H_{mk+1}$, and call them the \emph{pieces}.
We say a piece is \emph{moderate} if its number of vertices is between $n/(mf(n))$ and $nf(n)/m$.
Note that the number of vertices in a piece has expected value $\Theta(n/m)$.
The following lemma is proved by showing
this random variable does not deviate much from its expected value.

\begin{lemma}
\label{lem:goodsizes}
With high probability,
there are $o(m)$ non-moderate pieces.
\end{lemma}

\begin{proof}
We prove the first piece, $H_1$, is moderate whp.
By symmetry, this would imply that the average number of non-moderate pieces is $o(m)$.
By Markov's inequality, this gives that whp there are $o(m)$ non-moderate pieces.
Let $X$ denote the number of vertices of $H_1$.
Note that $X$ is distributed as $k + \polya(1,km,k,n-m)$; so its expected value is $k + \frac{n-m}{1+km} = \Theta(n/m)$.
By Markov's inequality, $\Pr{X > n f(n) / m} = o(1)$.

For bounding $\Pr{X < n / (m f(n))}$, 
we use an alternative way to define the random variable 
$\polya(1,km,k,n-m)$ (see \cite[page 181]{urn_models_application}): 
assume $Z$ is a beta random variable with parameters $1/k$ and $m$. 
Then $X-k$, which has the same distribution as $\polya(1,km,k,n-m)$, is distributed as a binomial random variable with parameters $n-m$ and $Z$.
Note that
\begin{align*}
\Pr{Z < 3 / (mf(n))} & = 
\frac{\Gamma(m+1/k)}{\Gamma(m)\Gamma(1/k)}
\int_{0}^{3 / (mf(n))}
x^{1/k-1}(1-x)^{m-1} dx \\
& <
\frac{m^{1/k}}{\Gamma(1/k)}
\int_{0}^{3 / (mf(n))}
x^{1/k-1} d x
=
\frac{3^{1/k} k}{\Gamma(1/k) f(n)^{1/k}} = o(1) \:,
\end{align*}
where we have used the fact
${\Gamma(m+1/k)} < {\Gamma(m)} m^{1/k}$
which follows from \cite[inequality (2.2)]{gamma_inequalities}.
On the other hand, the Chernoff bound (see, e.g., \cite[Theorem~4.2]{rand_algs}) gives
\begin{align*}
\Pr{X < n / (mf(n)) | Z \ge 3 / (mf(n))} & \le
\Pr{\operatorname{Bin}(n-m,3 / (mf(n))) < n / (mf(n)) } \\
& < \exp (-3(n-m)/(8mf(n))) = o(1) \:,
\end{align*}
thus $\Pr{X < n / (m f(n))}=o(1)$.
\end{proof}

\begin{proof}[Proof of Theorem~\ref{thm:lowerbound}]
Consider an alternative way to generate $G(n)$ from $G(m)$: first, we determine how many vertices each piece has, and then we expose the structure of the pieces.
Let $Y$ denote the number of moderate pieces.
By Lemma~\ref{lem:goodsizes} we have 
$Y = \Omega(m)$ whp.
We prove the theorem conditional on $Y = y$, where 
$y=\Omega(m)$ is otherwise arbitrary.
Note that after the sizes of the pieces are exposed, what happens inside each piece in rounds $m+1,\dots,n$ is mutually independent from other pieces.
Let $H$ be a moderate piece with $n_1$ vertices.
By Lemma~\ref{lem:barrier},
the probability that $H$ has an $\Omega(n_1^{1-1/k})$-barrier is $\Omega(n_1^{1/k-k})$.
Since $n/(mf(n)) \le n_1 \le nf(n)/m$, the probability that $H$ has a 
$\Omega((n/(mf(n))^{1-1/k})$-barrier is $\Omega((nf(n)/m)^{1/k-k})$.
Since there are $y=\Omega(m)$ moderate pieces in total, the probability that no moderate piece has an 
$\Omega\left((n/(mf(n)))^{1-1/k}\right)$-barrier
is at most 
$$(1 - \Omega((nf(n)/m)^{1/k-k}))^y 
\le \exp ( - \Omega(f(n))) = o(1)
\:,$$ 
so whp there exists an
$\Omega\left(n^{(k-1)/(k^2+k-1)}f(n)^{-2}\right)$-barrier in $G(n)$, as required.
\end{proof}

\section{Proof of Theorem~\ref{thm:apollonian}}
\label{k-apol}
In this section we analyze the {\sf Push-Pull} protocol on random  $k$-Apollonian networks. 
Since these networks are a sub-family of random $k$-trees, we can reuse the proof techniques in Section \ref{proofupper} to bound the time needed to inform almost all vertices. 
First, we formally define the random $k$-Apollonian process.

\begin{definition}[Random $k$-Apollonian process]\label{def_k_apol}
Let $k$ be a positive integer.
Build a sequence $A(0)$, $A(1),$ $\dots$ of random graphs as follows.
The graph $A(0)$ is just a clique on $k$ vertices.
This $k$-clique is marked as \emph{active}.
For each $1\leq t\leq n$, $A(t)$ is obtained from $A(t-1)$ as follows:
an active $k$-clique of $A(t-1)$ is chosen uniformly at random,
a new vertex is born and is joined to all vertices of the chosen $k$-clique.
The chosen $k$-clique is marked as \emph{non-active}, and all the new $k$-cliques are marked as active in $A(t)$.
The graph $A(n)$ is called a \emph{random $k$-Apollonian network ($k$-RAN)} on $n+k$ vertices.
\end{definition}

We first prove a counterpart of Lemma~\ref{lem:largedegree} for $k$-RANs.

\begin{lemma}
\label{lem:largedegree2}
Let $1\le j \le n$ and let $q$ be a positive integer.
Let $x$ denote the vertex born in round $j$.
Conditional on any $A(j)$, the probability that $x$
has degree greater than $k + q (n/j)^{(k-2)/(k-1)}$ in $A(n)$ is $\Oh\left(q\sqrt{q}\exp(-q)\right)$.
\end{lemma}

\begin{proof}
Let $X=\polya\left(k, (k-1)(j-1), \begin{bmatrix}k-1&0\\ 1&k-2 \end{bmatrix}, n-j\right)$.
An argument similar to the proof of  Proposition~\ref{pro:degs} gives that the degree of $x$ in $A(n)$ is  distributed as $k + \left(X-k\right)/(k-2)$.
By Proposition~\ref{pro:moments},
$$\Ex{X^q} \le
(1+o(1)) \left( \frac{(k-1) (n-j)}{(k-1)j+1} \right)^{\frac{q(k-2)}{k-1}}\
 \prod_{i=0}^{q-1} \left( k + i (k-2) \right)
\le \left( \frac{n}{j} \right)^{\frac{q(k-2)}{k-1}} (k-2)^{q} (q+1)!
 \:.
$$
Thus,
\begin{align*}
\Pr{\deg(x) > k + q (n/j)^{(k-2)/(k-1)}} &=
\Pr{X-k > (k-2) q (n/j)^{(k-2)/(k-1)}} \\
& \le \frac{\Ex{X^q}}{\left((k-2)q (n/j)^{(k-2)/(k-1)}\right)^q} \\
& \le  (q+1)!q^{-q} = \Oh\left(q\sqrt{q}\exp(-q)\right) \:. \qedhere
\end{align*}
\end{proof}

Fix $k>2$ and let $f(n)=o(\log \log n)$ be an arbitrary function going to infinity with $n$, and
let
{$$m = \left \lceil \frac{n}{(\log n)^{2/(k-1)}f(n)^{(2k-2)/(k^2-2k)}} \right \rceil\:.$$}
Finally, let $q = \lceil 4 \log \log n \rceil$ and let
\begin{equation}
\label{tau_def2}
\tau = 2k + q (n/m)^{(k-2)/(k-1)} \:.
\end{equation}

An argument similar to the proof of Lemma~\ref{lem:log2} gives that whp a $k$-RAN on $n+k$ vertices has diameter $\Oh(\log n)$.
Theorem~\ref{thm:apollonian} thus follows from Lemma~\ref{lem:feige} and the following structural result, which we prove in the rest of this section.

\begin{lemma}
\label{thm:structural2}
Let $A$ be an $(n+k)$-vertex $k$-RAN.
Whp there exists  $\Sigma \subseteq V(A)$ satisfying the conditions of Lemma~\ref{lem:feige}
with $\tau$ defined in (\ref{tau_def2}) and $\chi = \Oh(\log n + \operatorname{diam}(A))$.
\end{lemma}

The proof of Lemma~\ref{thm:structural2} is along the lines of that of Lemma~\ref{thm:structural}.
For the rest of this section, $A=A(n)$ is an $(n+k)$-vertex $k$-RAN.
Consider the graph $A(m)$, which has $k+m$ vertices and $m(k-1)+1$ active $k$-cliques.
For any edge $e$ of $A(m)$, let $\na(e)$ denote the number of active $k$-cliques of $A(m)$ containing  $e$. 
Note that, since $k>2$, for each edge $e$, the number of active $k$-cliques containing $e$ does not decrease as the $k$-RAN evolves.
We define a spanning forest $F$ of $A(m)$ as follows: 
at round 0, $F$ has $k$ isolated vertices, i.e.\ the vertices of $A(0)$;
then for every $1\le t\le m$,
if the vertex $x$ born in round $t$ is joined to the $k$-clique $C$,
then in $F$, $x$ is joined to a vertex $u \in V(C)$ such that
$$N^{\ast}(x u) = \max_{v\in V(C)}N^{\ast}(x v) \:.$$
Note that $F$ has $k$ trees and the $k$ vertices of $A(0)$ lie in distinct trees.
Let $\mathsf{LOG}$ denote the event `each tree in $F$ has height $\Oh(\log n)$.'

\begin{lemma}
\label{lem:log2}
With high probability $\mathsf{LOG}$ happens.
\end{lemma}
\begin{proof}
{We prove that whp every path 
$u_h u_{h-1} \cdots u_0$ in $A(n)$
such that $u_i$ is born later than $u_{i-1}$
for all $i$, has length $\Oh(\log n)$.}
The proof is very similar to that of Proposition~\ref{pro:diameter}, the only difference being that the built auxiliary tree is indeed a random $k$-ary recursive tree, whose height is whp $\Oh(\log n)$ by Theorem~\ref{pro:drmota}.
\end{proof}

We prove Lemma~\ref{thm:structural2} conditional on the event $\mathsf{LOG}$.
In fact, we prove it for any $A(m)$ that satisfies $\mathsf{LOG}$.
Let $A_1$ be an arbitrary instance of $A(m)$ that satisfies $\mathsf{LOG}$.
{So, $A_1$ and $F$ are fixed in the following, and}
all randomness refers to rounds $m+1,\dots,n$.
The following deterministic lemma
will be used in the proof of Lemma~\ref{FastEdges2}.

\begin{lemma}
\label{lem:counting_new2}
Assume that $x y \in E(F)$ and $x$ is born later than $y$.
If the degree of $x$ in $A_1$ is at least $2k-1$, then $\na(xy) \ge {(k-1)^2/2}$.
\end{lemma}

\begin{proof}
Assume that $x$ is joined to $u_1,\dots,u_k$ when it is born, and that $v_1,v_2,\dots,v_{k-1},\dots$ are the neighbors of $x$ that are born later than $x$,
in the order of birth.
Let $\Psi$ denote the number of pairs $(u_j,C)$, where $C$ is an active $k$-clique in $A_1$ with $x u_j \subseteq E(C)$.
Consider the round in which vertex $x$ is born and is joined to $u_1,\dots,u_k$.
{For every $j\in\{1,\dots,k\}$, the edge $xu_j$ is contained in $k-1$ new active $k$-cliques,
so in this round $\Psi$ increases by $k(k-1)$.}
For each $i\in\{1,\dots, k-1\}$, consider the round in which vertex $v_i$ is born.
At least $k-i$ of the $u_j$'s are joined to $v_i$ in this round.
Each vertex $u_j$ that is joined to $v_i$ in this round is contained in $k-2$ new $k$-cliques that contain $x$,
and one $k$-clique containing $u_j$ is deactivated.
Hence in this round $\Psi$ increases by at least $(k-i)(k-3)$.
Consequently, right after $v_{k-1}$ is born, we have
{$$\Psi \ge k(k-1) + \sum_{i=1}^{k-1} (k-i)(k-3) = (k-1)^2k/2 \:.$$}
By the pigeonhole principle, there exists some $\ell\in\{1,\dots, k\}$ such that
the edge $x u_{\ell}$ is contained in at least ${(k-1)^2/2}$ active $k$-cliques, and 
this completes the proof,
as the number of active $k$-cliques containing $x u_{\ell}$ will not decrease later.
\end{proof}

A vertex of $A$ is called \emph{modern} if it is born later than the end of round $m$,
and is called \emph{traditional} otherwise.
In other words, vertices of $A_1$ are traditional and vertices of $A-A_1$ are modern.
We say edge $uv\in E(A)$ is \emph{fast} if at least one of the following is true:
$\deg(u)\le \tau$, or $\deg(v)\le \tau$, or $u$ and $v$ have a common neighbor $w$ with $\deg(w)\le \tau$.
For an edge $uv \in E(F)$, let $p_S(uv)$ denote the probability that $uv$ is not fast,
and let $p_S$ denote the maximum of $p_S$ over all edges of $F$.

\begin{lemma}\label{FastEdges2}
We have $p_S  = o ( 1 / (f(n) \log n))$.
\end{lemma}

\begin{proof}
The proof is similar to that of Lemma~\ref{FastEdges}.
Let $xy \in E(F)$ be arbitrary.
By symmetry we may assume that $x$ is born later than $y$.
By Lemma~\ref{lem:counting_new2}, at least one of the
following is true: vertex $x$ has less than $2k-1$ neighbors in $A_1$,
or 
{$\na(xy) \ge (k-1)^2/2$}.
So we may consider two cases.
  \begin{itemize}
\item 
Case 1:  vertex $x$  has less than $2k-1$ neighbors in $A_1$.
In this case  $x$ lies in at most $k+(k-2)^2$ many active $k$-cliques of $A_1$.
Suppose  that $x$ has $D$ neighbors in $A_1$ and lies in $B$ many active $k$-cliques in $A_1$.
Let
$$X = \polya\left(B,(k-1)m+1-B,\begin{bmatrix}k-1&0\\ 1&k-2 \end{bmatrix},n-m\right)\:.$$
Then by an argument similar to the proof of Proposition~\ref{pro:degs}, the degree of $x$ is distributed as
$D + \left( X - B\right)/(k-2)$.
By Proposition~\ref{pro:moments},
\begin{align*}
\Ex{X^q} & \le
 (1+o(1))
 \left( \frac{(k-1) (n-m)}{(k-1)m+1} \right)^{\frac{q(k-2)}{k-1}}
 \prod_{i=0}^{q-1} \left( B + i (k-2) \right) \\
&\le\Oh\left( \left( \frac{n}{m} \right)^{\frac{q(k-2)}{k-1}} (k-2)^q (k+q)!   \right) \:,
\end{align*}
where we have used $B \le k(k-2)$.
Therefore,
\begingroup
\addtolength{\jot}{1em}
\begin{align*}
& \Pr{\deg(x) > 2k + q (n/m)^{\frac{k-2}{k-1}}}  \le \Pr{X \ge (k-2)q (n/m)^{\frac{k-2}{k-1}}} \\
& \le \frac{\Ex{X^q}}{(k-2)^qq^q (n/m)^{\frac{q(k-2)}{k-1}}} = \Oh \left( \frac{(k+q)!}{q^q} \right) = o \left( \frac{1}{f(n) \log n}\right) \:.
\end{align*}
\endgroup

\item Case 2: $\na(xy) \ge {(k-1)^2/2}$.
In this case we bound from below the probability that there exists a modern vertex $w$
that is adjacent to $x$ and $y$ and has degree at most  $\tau$.
We first bound from above the probability that $x$ and $y$ have no modern common neighbors.
For this to happen, none of the $k$-cliques containing $x$ and $y$ must be chosen in rounds $m+1,\dots,n$.
This probability equals $$
p := \Pr{\polya(\na(xy),m(k-1)+1-\na(xy),k-1,n-m)={\na(xy)}} \:.$$
Since $\na(xy) \ge {(k-1)^2/2}$, by Proposition~\ref{pro:polyasimple} we have
\begin{equation*}
p \le \left( \frac{m+1}{n}\right)^{{{(k-1)}}/{2}}
= o \left( \frac{1}{f(n) \log n}\right) \:.
\end{equation*}

Now, assume that $x$ and $y$ have a modern common neighbor $w$.
If there are multiple such vertices, choose the one that is born first.
Since $w$ appears later than round $m$, by
Lemma~\ref{lem:largedegree2},
$$\Pr{\deg(w) > k + q (n/m)^{(k-2)/(k-1)}} =\Oh\left( q\sqrt{q}\exp(-q) \right) =o \left( \frac{1}{f(n) \log n}\right) \:.\qedhere
$$
\end{itemize}
\end{proof}

Enumerate the $k$-cliques of $A_1$ as $C_1$, $C_2,$ $\dots,$ and $C_{m(k-1)+1}$.
Then choose $r_1\in C_1,\dots,$ $r_{m(k-1)+1}\in C_{m(k-1)+1}$ arbitrarily, and call them the \emph{representative vertices}.
Starting from $A_1$, when modern vertices are born in rounds $m+1,\dots,n$ until $A$ is formed,
every clique $C_i$ `grows' to a $k$-RAN with a random number of vertices, which is a subgraph of $A$.
Enumerate these subgraphs as $H_1,\dots,H_{m(k-1)+1}$, and call them the \emph{pieces}.
More formally, $H_1,\dots,H_{m(k-1)+1}$ are induced subgraphs of $A$ such that
a vertex $v$ is in $V(H_j)$ if and only if every path connecting $v$ to a traditional vertex intersects $V(C_j)$.

A traditional vertex is called \emph{nice} if it is connected to some vertex in $A(0)$ via a path of fast edges.
Since $F$ has height $\Oh(\log n)$ and each edge of $F$ is fast with probability at least $1-p_S$,
the probability that a given traditional vertex is not nice is $\Oh(p_S \log n)$ by the union bound.
A piece $H_j$ is called \emph{nice} if all its modern vertices have degrees at most $\tau$,
and the vertex $r_j$ is nice.
A modern vertex is called \emph{nice} if it lies in a nice piece.
A vertex/piece is called \emph{bad} if it is not nice.

\begin{lemma}
\label{lem:nonnice2}
The expected number of bad vertices is $o(n)$.
\end{lemma}

\begin{proof}
The proof is very similar to that of Lemma~\ref{lem:nonnice},
except we use 
Lemmas~\ref{lem:largedegree2}
and~\ref{FastEdges2}
instead of
Lemmas~\ref{lem:largedegree}
and~\ref{FastEdges}, respectively.
\end{proof}

The proof of Lemma~\ref{thm:structural2}
is exactly the same as that of Lemma~\ref{thm:structural},
except we use
Lemmas~\ref{FastEdges2} and~\ref{lem:nonnice2}
instead of
Lemmas~\ref{FastEdges} and~\ref{lem:nonnice},
respectively.
This concludes the proof of Theorem~\ref{thm:apollonian}.

\bibliographystyle{plain}
\bibliography{rumorspreading}

\end{document}